\documentclass[letterpaper,11pt]{article}
\usepackage[margin=1in]{geometry}
\usepackage{amsmath, amsfonts, amssymb, amsthm}
\usepackage{thm-restate}
\usepackage{bbm,url}
\usepackage[colorlinks=true,urlcolor=black,linkcolor=black,citecolor=black]{hyperref}
\usepackage{graphicx}
\usepackage{xcolor}
\usepackage{colortbl}
\usepackage{subcaption}
\usepackage{hhline}
\usepackage{multirow}
\usepackage{multicol}
\usepackage{pifont}
\usepackage{enumitem}
\usepackage{comment}
\usepackage{nicefrac,bm}
\usepackage{algorithm}
\usepackage[noend]{algpseudocode}
\usepackage{thmtools}
\usepackage{thm-restate}
\algtext*{Until}

\allowdisplaybreaks

\newtheorem{theorem}{Theorem}
\newtheorem*{theorem*}{Theorem}

\newtheorem{lemma}{Lemma}
\newtheorem{claim}{Claim}

\newtheorem{remark}{Remark}

\theoremstyle{definition}

\def \eps{\varepsilon}

\title{Dynamic Vector Bin Packing for \\Online Resource Allocation in the Cloud}

\author{
Aniket Murhekar\footnote{University of Illinois at Urbana-Champaign, USA}\\
\texttt{\small aniket2@illinois.edu}
\and
David Arbour\footnote{Adobe Research, USA}\\
\texttt{\small arbour@adobe.com}\\
\and
Tung Mai\footnote{Adobe Research, USA}\\
\texttt{\small tumai@adobe.com}\\
\and
Anup Rao\footnote{Adobe Research, USA}\\
\texttt{\small anuprao@adobe.com}
}

\AtBeginDocument{%
  \providecommand\BibTeX{{%
    \normalfont B\kern-0.5em{\scshape i\kern-0.25em b}\kern-0.8em\TeX}}}

\newcommand{\R}{\mathcal{R}}
\newcommand{\A}{\mathcal{A}}
\newcommand{\vect}[1]{\mathbf{#1}}
\newcommand{\size}{\mathbf{s}}
\newcommand{\cost}{\mathsf{cost}}
\renewcommand{\span}{\mathsf{span}}
\newcommand{\opt}{\mathsf{OPT}}
\newcommand\norm[1]{\lVert#1\rVert}

\allowdisplaybreaks
\begin{document}

\maketitle

\begin{abstract}
Several cloud-based applications, such as cloud gaming, rent servers to execute jobs which arrive in an online fashion. Each job has a resource demand, such as GPU requirement, and must be dispatched to a cloud server which has enough resources to execute the job, which departs after its completion. Under the ``pay-as-you-go'' billing model, the server rental cost is proportional to the total time that servers are actively running jobs. The problem of efficiently allocating a sequence of online jobs to servers without exceeding the resource capacity of any server while minimizing total server usage time can be modelled as a variant of the dynamic bin packing problem (DBP), called MinUsageTime DBP \cite{LiTangCai2014}. 

In this work, we initiate the study of the problem with multi-dimensional resource demands (e.g. CPU/GPU usage, memory requirement, bandwidth usage, etc.), called MinUsageTime Dynamic Vector Bin Packing (DVBP). We study the competitive ratio (CR) of Any Fit packing algorithms for this problem. We show almost-tight bounds on the CR of three specific Any Fit packing algorithms, namely First Fit, Next Fit, and Move To Front. We prove that the CR of Move To Front is at most $(2\mu+1)d +1$, where $\mu$ is the ratio of the max/min item durations. For $d=1$, this implies a significant improvement over the previously known upper bound of $6\mu+7$ \cite{KamaliLopezOrtiz2015}. We then prove the CR of First Fit and Next Fit are bounded by $(\mu+2)d+1$ and $2\mu d+1$, respectively. Next, we prove a lower bound of $(\mu+1)d$ on the CR of any Any Fit packing algorithm, an improved lower bound of $2\mu d$ for Next Fit, and a lower bound of $2\mu$ for Move To Front in the 1-D case. All our bounds improve or match the best-known bounds for the 1-D case. Finally, we experimentally study the average-case performance of these algorithms on randomly generated synthetic data, and observe that Move To Front outperforms other Any Fit packing algorithms.
\end{abstract}
\maketitle

\section{Introduction}
Bin packing is an extensively studied problem in combinatorial optimization \cite{coffman2013BPsurvey}. The goal of the classical bin packing problem is to pack a given set of items with different sizes into the smallest number of identical bins such that the total size of items in each bin does not exceed the capacity of the bin. The dynamic bin packing problem (DBP) \cite{coffman1983dbp} is a generalization of the classical bin packing problem, where items can arrive and depart over time, and the objective is to minimize the number of bins used over time. Dynamic bin packing naturally models several resource allocation problems, including those arising in cloud computing \cite{dbp-app-2,dbp-app-1}.

Motivated by cloud computing applications where the goal is to dispatch jobs arriving in an online fashion to servers, with the objective of minimizing the server usage time, Li, Tang, and Cai \cite{LiTangCai2014} introduced a variant of dynamic bin packing called \textit{MinUsageTime Dynamic Bin Packing}. In this variant, items appear in an online fashion and must be packed into resource-bounded bins. When an item (job) arrives, it must immediately be dispatched to a bin (server) which has enough resources to accommodate (execute) the job. The objective is to minimize the \textit{total time} that bins are \textit{active}, i.e., contain at least one active item that has not yet departed. Moreover, due to overheads involved in migrating jobs from one server to another, it is assumed that the placement of an item to a bin is irrevocable. The objective function, the total usage time of the bins, naturally models the power consumption or rental cost of the servers. Below we discuss two concrete applications motivating the MinUsageTime Dynamic Bin Packing problem, one faced by cloud service provider and the other by the cloud service user.

\paragraph{Virtual machine placement on physical servers.} A popular way that cloud resource providers offer their services to users is through the use of Virtual Machines (VMs). Users can request VMs with certain resource demands, and in turn cloud resource managers place these VMs on physical servers with sufficient resource capacity to serve the VM requests. Minimizing the total usage time of the physical machines can directly lead to power and cost savings on the cloud provider end \cite{buchbinder21vm,dynamic-right-sizing}. As \cite{hadary2020protean} suggests, even a 1\% improvement in packing efficiency can lead to cost savings of roughly \$100 million per year for Microsoft Azure. By viewing the VM requests as items and the physical servers as bins, the problem of minimizing the usage time of physical machines therefore directly translates to the MinUsageTime DBP problem. 

\paragraph{Cloud gaming and other cloud user applications.} Several organizations offer their services to customers by renting servers (as VMs) from on-demand public cloud providers such as Amazon EC2. They are typically charged according to their server usage times in hourly or monthly basis following the ``pay-as-you-go'' billing model \cite{AmazonEC}. Minimizing the organization's server renting cost is therefore equivalent to minimizing the usage time of the rented servers, thus reducing to the MinUsageTime DBP problem where customer jobs are viewed as items and rented servers as bins \cite{LiTangCai2014,KamaliLopezOrtiz2015,TangLiRenCai2016,LiTangCai2016,RenTangLiCai2017,Azar2019clairvoyant}. Organizations such as GaiKai \cite{gaikai}, OnLive\cite{onlive}, and StreamMyGame \cite{streammygame} offer cloud based gaming services where computer games run on rented cloud servers, thereby saving players from the overheads involved in set-up and maintenance of the hardware/software infrastructure required for the game. A request from a customer to play a game is dispatched to a gaming server which has enough resources such as GPU or bandwidth to run the game instance, which runs until the customer stops playing the game. In this context, the gaming service providers can greatly benefit by employing efficient algorithms that dispatch customers' game requests to rented servers minimize the server rental cost.\\~

{MinUsageTime Dynamic Bin Packing} is therefore a problem of commercial and industrial importance, and has consequently also received theoretical interest in recent years to analyze the performance of online algorithms for the problem \cite{LiTangCai2014,KamaliLopezOrtiz2015,TangLiRenCai2016,LiTangCai2016,RenTangLiCai2017,RenTang2016,Azar2019clairvoyant,buchbinder21vm}. The performance of an online algorithm is usually measured in terms of its \textit{competitive ratio} (CR) \cite{competitiveratio}, which is the worst-case ratio between the quality (e.g. total server renting cost) of algorithm's solution to the quality of the solution produced by an optimal, offline algorithm. In this paper, we study the \textit{non-clairvoyant} version of the problem, wherein the departure time of an item is unknown upon its arrival. In the context of cloud gaming, this models customers being able to play games for durations unknown to the cloud gaming service.  

Existing work has primarily focused on Any Fit packing algorithms, which is a well-studied family of algorithms for the classical bin packing problem. An Any Fit packing algorithm is an algorithm that opens a new bin only when an incoming item cannot be packed in any of the existing open bins. Any Fit packing algorithms are useful and well-studied because they take decisions based on the current system state and not its history, leading to a desirable simplicity in implementation and explainability, and a low computational and memory footprint. 

Li, Tang, and Cai \cite{LiTangCai2014,LiTangCai2016} showed that the competitive ratio of any Any Fit packing algorithm for the MinUsageTime DBP problem is at least $\mu+1$, where $\mu$ is the ratio of the max/min item durations. A series of works \cite{LiTangCai2014,LiTangCai2016,TangLiRenCai2016,RenTangLiCai2017} showed that the competitive ratio of First Fit, a specific Any Fit packing algorithm which tries to pack a new item into the earliest opened bin that can accommodate the item, is at most $\mu+3$. Likewise, Next Fit, which keeps only one open bin at a time to pack items, was shown to have a CR of at least $2\mu$ \cite{TangLiRenCai2016,RenTangLiCai2017} and at most $2\mu+1$ \cite{KamaliLopezOrtiz2015}. On the other hand, Best Fit, which tries to pack a new item into bin with highest load, was shown to have an unbounded CR \cite{LiTangCai2016}. Kamali and L\'{o}pez-Ortiz \cite{KamaliLopezOrtiz2015} studied another Any Fit packing algorithm called Move To Front, which tries to pack a new item into the bin which was most recently used. They showed that Move To Front has an (asymptotic) competitive ratio of at most $6\mu+7$, and conjectured that the CR is at most $2\mu+1$. They also performed an average-case experimental study of these algorithms, and found that Move To Front had the best average-case performance, closely followed by First Fit and Best Fit. 

\paragraph{Modelling Multi-dimensional Resources Demands.} All of the above previous works assumed the item sizes to be one-dimensional. They assume items/jobs have a single dominant resource, such as CPU or GPU demand. However, in practice, the resources demands of an item/job such as a VM request or a game instance are multi-dimensional, e.g., CPU and GPU usage, memory requirement, bandwidth usage, etc. In the bin packing literature, the multi-dimensional version is a problem of great significance and is extensively studied \cite{multidimbpsurvey, panigrahy2011heuristics}. The multidimensional nature of demand usually makes the problem much more challenging \cite{woeginger1997}. In this work, we study the generalization of the MinUsageTime DBP problem called \textit{MinUsageTime Dynamic Vector Bin Packing} (DVBP) where the sizes of items and bins are $d$-dimensional vectors. The design of online algorithms for DVBP and the analysis of their competitive ratios is therefore a natural and practically important problem, and was indicated as an important direction for future work by previous papers \cite{RenTang2016,TangLiRenCai2016,RenTangLiCai2017}. 

\subsection{Our Contributions}
In this work, we initiate the study of the multi-dimensional version of the MinUsageTime DBP problem, called MinUsageTime Dynamic Vector Bin Packing (henceforth referred to simply as DVBP), where item and bin sizes are $d$-dimensional vectors. We analyze the competitive ratios of Any Fit packing algorithms for the problem, including four specific algorithms: First Fit, Next Fit, Best Fit and Move To Front. Table~\ref{tab:results} summarizes the best known bounds on the CR of these algorithms, and contrasts our results with previous work. Our contributions are summarized below. 

\begin{itemize}[leftmargin=*]
\item We prove an upper bound of $(2\mu+1) d +1$ on the competitive ratio of Move To Front for DVBP. For $d=1$, this implies a significant improvement on the previously known upper bound of $6\mu+7$ shown by Kamali and L\'{o}pez-Ortiz~\cite{KamaliLopezOrtiz2015} to $2\mu+2$, and nearly settles their conjecture of the CR being $2\mu+1$. Central to our result is a novel decomposition of the usage periods of the bins used by Move To Front into two classes of intervals, and carefully analyzing the cumulative cost of intervals in each class.
\item We prove an upper bound of $(\mu+2)d+1$ on the competitive ratio of First Fit, and of $2\mu d+1$ for Next Fit. These results rely on new lower bounds on the cost of the optimum solution for the $d$-D case. Our upper bounds then follow by combining these bounds with analysis techniques inspired from upper bound results for the 1-D case \cite{RenTangLiCai2017,KamaliLopezOrtiz2015}. Note that the competitive ratio of Best Fit is unbounded even for the 1-D case \cite{LiTangCai2016}. 
\item We prove a lower bound of $(\mu+1)d$ on the competitive ratio of any Any Fit packing algorithm for DVBP. We also show a lower bound of $2\mu d$ on the CR of Next Fit and of $\max\{2\mu,(\mu+1)d\}$ for Move To Front. In conjunction with our upper bound results, these results show almost-tightness for the CR of First Fit and Next Fit for the $d$-D case, and of Move To Front for the 1-D case. Our results improve or match all known lower bounds for the 1-D case \cite{LiTangCai2014,LiTangCai2016,RenTangLiCai2017,TangLiRenCai2016,KamaliLopezOrtiz2015}. Due to the multi-dimensionality of the problem, lower bound results of the 1-D case do not directly translate to the $d$-D case, and hence we design new constructions to establish the lower bounds. 

At a high level, our constructions use carefully-designed sequences of items which force an Any Fit algorithm to open $\Omega(d\cdot k)$ bins for a parameter $k$, each of which contain an item of small size but long-duration, thus leading to a cost of $\approx \mu$ per bin. The optimal solution however packs all the small items into a single bin with cost $\approx \mu$ and the other items into $k$ bins with cost $\approx 1$, resulting in a total cost of $O(k + \mu)$, thus implying CR of $\Omega(\mu d)$. 
\item We perform an average-case experimental study of these algorithms on randomly generated synthetic data. We observe that Move To Front outperforms other Any Fit packing algorithms, with First Fit and Best Fit also performing well on average. 
\end{itemize}

\begin{table*}[!t]
\centering
\begin{tabular}{|p{0.15\textwidth}||p{0.16\textwidth}|p{0.16\textwidth}||p{0.18\textwidth}|p{0.18\textwidth}|}\hline
\textbf{Algorithm} & \textbf{Lower Bound $(d=1)$} & \textbf{Upper Bound $(d=1)$} & \textbf{Lower Bound $(d\ge 1)$} & \textbf{Upper Bound $(d\ge 1)$} \\\hhline{|=#=|=#=|=|}
Any Fit & $\mu+1$ \cite{LiTangCai2016,RenTangLiCai2017} & $\infty$ & \cellcolor{blue!25} $(\mu+1)d$ (Thm.~\ref{thm:lb-anyfit}) & $\infty$ \\\hline
Move To Front & \cellcolor{blue!25} $2\mu$ (Thm.~\ref{thm:lb-mtf}) & \cellcolor{blue!25} $2\mu + 2$ (Thm.~\ref{thm:ub-mtf}), improves \cite{KamaliLopezOrtiz2015} & \cellcolor{blue!25} $\max\{2\mu,(\mu+1)d\}$ (Thm.~\ref{thm:lb-mtf})& \cellcolor{blue!25} $(2\mu +1) d + 1~~~~~$ (Thm.~\ref{thm:ub-mtf}) \\\hline
First Fit & $\mu+1$ \cite{LiTangCai2016,RenTangLiCai2017} & $\mu+3$ \cite{RenTangLiCai2017} & \cellcolor{blue!25} $(\mu+1)d$ (Thm.~\ref{thm:lb-anyfit}) & \cellcolor{blue!25} $(\mu+2)d+1~~~~~$ (Thm.~\ref{thm:ub-ff})\\\hline
Next Fit & $2\mu$ \cite{TangLiRenCai2016} & $2\mu+1$ \cite{KamaliLopezOrtiz2015} & \cellcolor{blue!25} $2\mu d$ (Thm.~\ref{thm:lb-nextfit})& \cellcolor{blue!25} $2\mu d + 1$ (Thm.~\ref{thm:ub-nextfit})\\\hline
Best Fit & Unbounded \cite{LiTangCai2016} & $\infty$ & Unbounded \cite{LiTangCai2016} & $\infty$ \\\hline
\end{tabular}
\caption{\normalfont Summary of the best known upper and lower bounds on the competitive ratio of algorithms for the MinUsageTime Dynamic Vector Bin Packing problem in $d$ dimensions. $\mu$ denotes the ratio of max/min item durations. Colored cells highlight our results.}\label{tab:results}
\end{table*}

Given its bounded competitive ratio indicating good performance against adversarial examples, as well as good average-case performance, our theoretical and experimental results lead us to concur with the recommendation of \cite{KamaliLopezOrtiz2015} that Move To Front is the algorithm of choice for practical solutions to the DVBP problem, even in higher dimensions.

\subsection{Further Related Work}
Classical bin packing is known to be NP-hard even in the offline case \cite{gareyjohnson}. There is extensive work on designing algorithms with good competitive ratio for online versions of this problem \cite{Johnson1973NearoptimalBP,coffman07bp,coffman2013BPsurvey}, with 1.54037 and 1.58889 being the best-known lower and upper bounds \cite{bplb,bpub}. In online vector bin packing, the item sizes are $d$-dimensional vectors. Garey et al. \cite{GAREYdbp} showed that a generalization of First Fit has a CR of $d+0.7$, and Azar et al.~\cite{azarlb} showed an information-theoretic lower bound of $\Omega(d^{1-\eps})$. For further results on multi-dimensional versions of bin packing, we refer the reader to the survey \cite{multidimbpsurvey}. On the practical side, Panigrahy et al. \cite{panigrahy2011heuristics} studied heuristics for the offline vector bin packing problem. 

Dynamic bin packing with the objective of minimizing the number of bins is also the subject of several works \cite{coffman1983dbp,dbp1,dbp2,dbp3}. Coffman et al. \cite{coffman1983dbp} showed that First Fit has a competitive ratio of between 2.75 to 2.897, and Wong et al. \cite{wong_dbp_lb} showed a lower bound of 2.667 on the CR of any online algorithm. A further generalization called the fully dynamic bin packing problem, in which already packed items can be moved to different bins, has also been studied in \cite{dbp3}.

The MinUsageTime dynamic bin packing problem has been studied in several recent works \cite{LiTangCai2014,KamaliLopezOrtiz2015,TangLiRenCai2016,LiTangCai2016,RenTangLiCai2017,RenTang2016,Azar2019clairvoyant,buchbinder21vm}; Table~\ref{tab:results} cites the relevant prior work on the non-clairvoyant version of the problem. In the clairvoyant version of the problem the departure time of an item is known when it arrives \cite{RenTang2016,Azar2019clairvoyant}. This problem is known to have an algorithm with a $O(\sqrt{\log \mu})$ competitive ratio, with a matching lower bound \cite{Azar2019clairvoyant}. The interval scheduling problem \cite{intervalscheduling} is also closely related; see \cite{RenTang2016,buchbinder21vm} and references therein. In the presence of additional information about future load, algorithms with improved CR were presented by \cite{buchbinder21vm}. To the best of our knowledge, the multi-dimensional version of the MinUsageTime DBP problem has not been studied, though it finds mention as a direction for future work in \cite{RenTang2016,TangLiRenCai2016,RenTangLiCai2017}.

\paragraph{Organization.} The rest of the paper is organized as follows. Section~\ref{sec:prelim} introduces notation, relevant definitions, packing algorithms, and useful preliminary observations. Section~\ref{sec:mtf}, \ref{sec:ff}, and \ref{sec:nf} establish upper bounds on the competitive ratios of Move To Front, First Fit, and Next Fit, respectively. Section~\ref{sec:lb} presents lower bounds on competitive ratio of any Any Fit packing algorithm and certain improved lower bounds for specific algorithms. Section~\ref{sec:experiments} discusses our experimental results examining the average-case performance of various Any Fit packing algorithms on randomly generated synthetic data. Finally, some concluding remarks and directions for future work are presented in Section~\ref{sec:conclusion}.

\section{Notation and Preliminaries}\label{sec:prelim}
For $n\in\mathbb{N}$, let $[n]$ denote the set $\{1,2,\dots, n\}$. The $L_\infty$ norm of a vector a vector $\vect{v}\in \mathbb{R}_{\ge 0}^d$ is denoted by $\norm{\vect{v}}_\infty$ and equals $\max_{i\in [d]}\vect{v}_j$. We will use the following simple properties of the $L_\infty$ norm, which are proved in Appendix~\ref{app:prelim} for completeness.
\begin{restatable}{proposition}{propnorm}\label{prop:norm}
The $L_\infty$ norm satisfies the following.
\begin{enumerate}
\item[(i)] For a vector $\vect{v}\in\mathbb{R}_{\ge 0}^d$ and a constant $c\ge 0$, $\norm{c\cdot\vect{v}}_\infty = c\cdot \norm{\vect{v}}_\infty$.
\item[(ii)] For any set of vectors $\vect{v}_1,\dots,\vect{v}_n \in \mathbb{R}_{\ge 0}^d$, we have: \[ \bigg\lVert{\sum_{i=1}^n \vect{v}_i}\bigg\rVert_\infty \le \sum_{i=1 }^n\norm{\vect{v}_i}_\infty \le d\cdot \bigg\lVert{\sum_{i=1}^n \vect{v}_i}\bigg\rVert_\infty.\]
\end{enumerate}
\end{restatable}

\subsection{Problem Definition} We now formally define the online MinUsageTime Dynamic Vector Bin Packing (DVBP) problem.
\paragraph{\bf Problem Instance.}  Let $d\in \mathbb{N}$ denote the number of resource dimensions, i.e., CPU, memory, I/O, etc. We let $\R$ denote the list of items. Each item $r\in \R$ is specified by a tuple $(a(r),e(r),\size(r))$, where $a(r), e(r)\in\mathbb{Q}_{\ge 0}$ and $\size(r)$  denote the arrival time, departure time, and the size of the item, respectively. Note that each item has multi-dimensional resource demands, i.e., $\size(r) \in \mathbb{R}_{\ge 0}^d$ where $\size(r)_j$ denotes the size of the item in the $j^{th}$ dimension, for $j\in[d]$. Without loss of generality, we assume that bins have unit capacity in each dimension, i.e., the size of a bin is $\vect{1}^d$ and that $\size(r)\in[0,1]^d$ for each $r\in\R$ by normalization. Further, let $\size(\R) = \sum_{r\in\R} \size(r)$.

For an item $r\in\R$, let $I(r) = [a(r), e(r))$ denote the \textit{active interval}\footnote{For technical reasons $I(r)$ is half open, i.e., the item $r$ has departed at time $e(r)$.} of item $r$, and we say that item $r$ is \textit{active} in the interval $I(r)$. Let $\ell(I(r)) = e(r) - a(r)$ denote the length of interval $I(r)$, i.e., the \textit{duration} of item $r$. W.l.o.g, we assume $\min_{r\in \R} \ell(I(r)) = 1$, and define $\mu := \max_{r\in \R} \ell(I(r))$. Thus, $\mu$ denotes the ratio of the max/min item durations. Finally, let $\span(\R) = \ell(\cup_{r\in\R} I(r))$ denote the total length of time for which at least one item of $\R$ is active.

\paragraph{\bf Problem Objective.} We focus on the non-clairvoyant setting without recourse. This means that an online algorithm must pack an item immediately into a single bin when it arrives, and that the algorithm cannot repack items. Moreover, when an item arrives the algorithm does not have any knowledge of when it will depart. Let $P_{\A,\R}$ denote the \textit{packing} of the items $\R$ by the algorithm $\A$. Let $B_1,\dots,B_m$ be the bins opened by $\A$, and let $R_i$ be the items placed on bin $B_i$. We assume the cost of using a bin for an interval $I$ equals its length $\ell(I)$. Then the cost of the packing $P_{\A,\R}$ is defined as the total usage time of all the bins, i.e.,
\begin{equation}\label{eq:usage-time}
    \cost(\A,\R) = \sum_{i=1}^m \span(R_i).
\end{equation}


\noindent With this problem objective, our goal is to compute a packing of $\R$ that minimizes the above cost.  

An empty bin is \textit{opened} the first time it receives an item, and remains \textit{open} as long as it contains an active item. When an open bin becomes empty, i.e., all items packed in it depart, we say that it is \textit{closed}. We can assume that once a bin is closed, it is never opened again, i.e., no item is packed in it again. This assumption is justified because bins are indistinguishable, and an idle bin has zero cost. Thus, a bin which has two usage periods $[a,b)$ and $[c,d)$ separated by an idle period $[b,c)$ can be replaced by two bins active between $[a,b)$ and $[c,d)$ respectively, without any change in the cost. Thus we can assume that the usage period of each bin is a single interval. Likewise, we assume that $\cup_{r\in \R} I(r)$ equals the single interval $[0, \span(\R))$, otherwise we can consider each interval of $\cup_{r\in \R} I(r)$ as a separate sub-problem.

\subsection{Any Fit Packing Algorithms}
We now discuss the Any Fit family of algorithms, which are adaptations of standard bin packing algorithms to the DVBP problem. An Any Fit packing algorithm maintains a list $L$ of open bins, and does not open a new bin upon the arrival of an item $r$, if $r$ can be packed into an open bin in $L$. Its pseudocode is given in Algorithm~\ref{alg:anyfit}. 

Different Any Fit packing algorithms differ in how an open bin $b\in L$ is selected to accommodate an item $r$ (Line 4), and how the list $L$ is modified (Lines 9 and 12). In this work, we focus on the following four Any Fit packing algorithms:

\begin{itemize}[leftmargin=*]
\item \textit{Move To Front}. Bins in the list $L$ are maintained in order of most-recent usage. Thus, when an item $r$ arrives, $r$ is placed in the bin $b$ which appears earliest in $L$ and can accommodate $r$, else a new bin $b$ is opened. Immediately $b$ is moved to the front of the list $L$ as it is the most recently used bin.
\item \textit{First Fit}. Bins in the list $L$ are maintained in increasing order of opening time. Thus, an item $r$ is placed in the earliest open bin that can hold $r$.
\item \textit{Next Fit}. At any given time, $|L| = 1$, i.e., at each time Next Fit maintains one open bin in $L$ as a designated \textit{current} bin. When an item $r$ does not fit into the current bin, the current bin is \textit{released} and a new bin is opened to pack $r$ and is made the current bin. 
\item \textit{Best Fit}. An item $r$ is placed in the ``most-loaded" bin. When $d=1$, the load of a bin containing a set of items $R$ is simply $\size(R)$. For $d\ge 2$, there is no unique way of computing the load $w(R)$ of set $R$ from the load vector $\size(R)$. A few options are:
\begin{itemize}
    \item Max load, i.e., $w(R) = \norm{\size(R)}_\infty$,
    \item Sum of loads, i.e., $w(R) = \norm{\size(R)}_1$,
    \item $L_p$-norm of the load, i.e., $w(R) = \norm{\size(R)}_p$, for $p\ge 2$.
\end{itemize}
\end{itemize}

\begin{algorithm}[t]\caption{Any Fit Packing Algorithm}\label{alg:anyfit}
\begin{algorithmic}[1]
\State $L \gets \{b\}$, where $b$ is an empty open bin \Comment{List of open bins}
\Repeat{\bf :}
\If{an item $r$ arrives}
\If{$r$ fits in an open bin of $L$} 
\State Choose a bin $b\in L$ that can hold $r$
\Else 
\State Open a new bin $b$ and add $b$ to $L$
\EndIf
\State Pack item $r$ in bin $b$
\State Modify $L$ as needed
\EndIf
\If{an item $r$ departs}
\State Remove closed bins from $L$
\State Modify $L$ as needed
\EndIf
\Until
\end{algorithmic}
\end{algorithm}

\paragraph{Competitive Ratio.} We measure the performance of an online algorithm $\A$ by its \textit{competitive ratio}, i.e., the worst-case ratio between the cost of the packing produced by algorithm $\A$ and the cost of the packing produced by the optimal offline algorithm which can repack items \cite{competitiveratio}. For a list of items $\R$, denote the optimal, offline cost by $\opt(\R)$. An algorithm $\A$ is said to be $\alpha$-competitive (for $\alpha \ge 1$) if for all item lists $\R$, we have $\cost(\A,\R) \le \alpha\cdot\opt(\R)$. Naturally, we desire algorithms where $\alpha $ is as small as possible.

\subsection{Lower bounds on the Optimum Cost} To analyze the competitive ratio of online algorithms, it is useful to place lower bounds on the optimum cost. To this end, let $\size(\R,t) = \sum_{r\in\R:t\in I(r)} \size(r)$ denote the total size of items that are active at time $t$. Let $\opt(\R,t)$ denote the number of bins the optimal offline algorithm has open at time $t$, equivalently, it is the smallest number of bins into which all items active at time $t$ can be repacked. Then:
\begin{equation}\label{eq:opt-objective}
    \opt(\R) = \int_{\min_{r\in\R} a(r)}^{\max_{r\in\R} e(r)} \opt(\R,t) dt.
\end{equation}

The following lemma presents $d$-dimensional generalizations of lower bounds on OPT introduced in earlier works \cite{LiTangCai2016,RenTangLiCai2017}.
\begin{lemma}\label{lem:lb-opt} The following are lower bounds on $\opt(\R)$.
\begin{enumerate}
\item[(i)] $\opt(\R) \ge \int_{\min_{r\in\R} a(r)}^{\max_{r\in\R} e(r)} \lceil\norm{\size(\R,t)}_\infty\rceil \: dt$
\item[(ii)] $\opt(\R) \ge \frac{1}{d} \sum_{r\in R} \norm{\size(r)}_\infty \cdot \ell(I(r))$
\item[(iii)] $\opt(\R) \ge \span(\R)$
\end{enumerate}
\end{lemma}
\begin{proof}
The definition of $\size(\R,t)$ and the size of bins being $\vect{1}^d$ implies that any algorithm needs at least $\lceil \size(\R,t)_j \rceil$ bins to pack the total load on the $j^{th}$ dimension, for any $j\in [d]$. Thus, $\opt(\R,t) \ge \max_{j\in [d]} \lceil \size(\R,t)_j \rceil = \lceil\norm{\size(\R,t)}_\infty\rceil$. Using \eqref{eq:opt-objective}, we obtain (i). 

Define the \textit{time-space utilization} of an item $r$ as $u(r) = \norm{\size(r)}_\infty\cdot \ell(I(r))$. The following shows that the total time-space utilization of all items is a lower bound on $d\cdot \opt$, thus proving (ii).
\begin{equation*}
\begin{aligned}
\opt(\R) &\ge \int_{\min_{r\in\R} a(r)}^{\max_{r\in\R} e(r)} \norm{\size(\R,t)}_\infty \: dt & \text{(using (i))}\\
&\ge \int_{\min_{r\in\R} a(r)}^{\max_{r\in\R} e(r)} \bigg\lVert{\sum_{r: t\in I(r)}\size(r)}\bigg\rVert_\infty dt & \text{(def of $\size(\R,t)$)}\\
&\ge \frac{1}{d} \int_{\min_{r\in\R} a(r)}^{\max_{r\in\R} e(r)}\sum_{r: t\in I(r)} \norm{\size(r)}_\infty \: dt & \text{(using Prop~\ref{prop:norm})}\\
&= \frac{1}{d}\sum_{r\in\R} \norm{\size(r)}_\infty\cdot \ell(I(r)). &\text{(swap order)}
\end{aligned}
\end{equation*}

Lastly, observe that since at least one bin is needed for each time $t$ that an item is active, we have $\opt(\R,t)\ge 1$ for each time instant $t\in [0,\span(\R))$. Together with \eqref{eq:opt-objective}, this implies (iii). 

Note that the lower bound (i) is tighter than both (ii) and (iii).
\end{proof}

\section{Upper Bound on the Competitive Ratio of Move To Front}\label{sec:mtf}
In this section we prove the first main result of our paper. 
\begin{theorem}\label{thm:ub-mtf}
The competitive ratio of Move To Front for the MinUsageTime Dynamic Vector Bin packing problem in $d$-dimensions is at most $(2\mu+1)d +1$. 
\end{theorem}

For $d=1$, our result implies that Move To Front has a competitive ratio of at most $2\mu+2$. This significantly improves the result of Kamali and L\'{o}pez-Ortiz~\cite{KamaliLopezOrtiz2015}, who showed that Move To Front has an \textit{asymptotic} competitive ratio of $6\mu + 7$, i.e., for any item list $\R$, they showed $\cost(\text{MF},\R) \le (6\mu+7)\cdot\opt(\R) + 3(\mu+1)$. Our result also nearly settles their conjecture of the CR being $2\mu+1$. Their analysis decomposes the active span into segments of length $(\mu+1)$ and compares the cost of OPT with the cost of Move To Front in each such interval. It turns out that this decomposition is sub-optimal. Instead, we directly use the nature of the Move To Front algorithm and develop a novel decomposition of the usage periods of each bin $B$ into intervals based on whether or not in the interval $B$ is the most recently used bin. We now prove Theorem~\ref{thm:ub-mtf}.

Suppose Move To Front uses $m$ bins $B_1,B_2,\dots,B_m$ on an input sequence $\R$. As mentioned earlier, we can assume that $\cup_{r\in\R} I(r) = [0,\span(\R))$ and that the usage period of each bin is an interval For $i\in[m]$, let $I_i = \span(R_i)$ denote the usage period/active interval of bin $B_i$, where $R_i$ is the set of items packed in $B_i$. The cost of Move To Front (MF) can be expressed as $\cost(MF,\R) = \sum_{i=1}^m \ell(I_i)$.

Recall that Move To Front maintains a list $L$ of open bins in the order of their most-recent usage. We say a bin is a \textit{leader} at time $t$ if it is in the front of the list $L$ at time $t$. We call an interval $I$ a \textit{leading interval} for bin $B$ if $B$ is a leader at every time instant in $I$. If Move To Front packs an item into a bin $B$, then $B$ is immediately made the leader. Thus, if a bin $B$ is not a leader at time $t$, then it cannot accept a new item at $t$. Based on the above definition, we partition the active interval of each bin $B$ into intervals which alternate between leading intervals for $B$ and non-leading intervals for $B$. Clearly the time at which a bin is opened begins a leading period for the bin. Thus, for each $i\in[m]$, the interval $I_i$ is sequentially partitioned into $2n_i$ (half-open) intervals as $I_i = P_{i,1} \cup Q_{i,1} \cup P_{i,2} \cup Q_{i,2} \cup \cdots \cup P_{i,n_i} \cup Q_{i,n_i}$, where each $P_{i,j}$ is a leading interval for bin $B_i$ and $Q_{i,j}$ is a non-leading interval for $j\in[n_i]$. Since empty intervals have zero cost, we can assume that all intervals except perhaps the last non-leading intervals of each bin are non-empty, i.e., perhaps $Q_{i,n_i} =\emptyset$. This decomposition is illustrated in Figure~\ref{fig:mtf} with red/thick lines representing leading intervals and blue/thin lines representing non-leading intervals. Using this decomposition, one can write the cost as:
\begin{equation}\label{eq:mtf-cost}
\cost(MF,\R) = \sum_{i=1}^m \sum_{j=1}^{n_i} \bigg(\ell(P_{i,j}) + \ell(Q_{i,j})\bigg).
\end{equation}

\begin{figure}[t]
\centering
\includegraphics[scale=0.5]{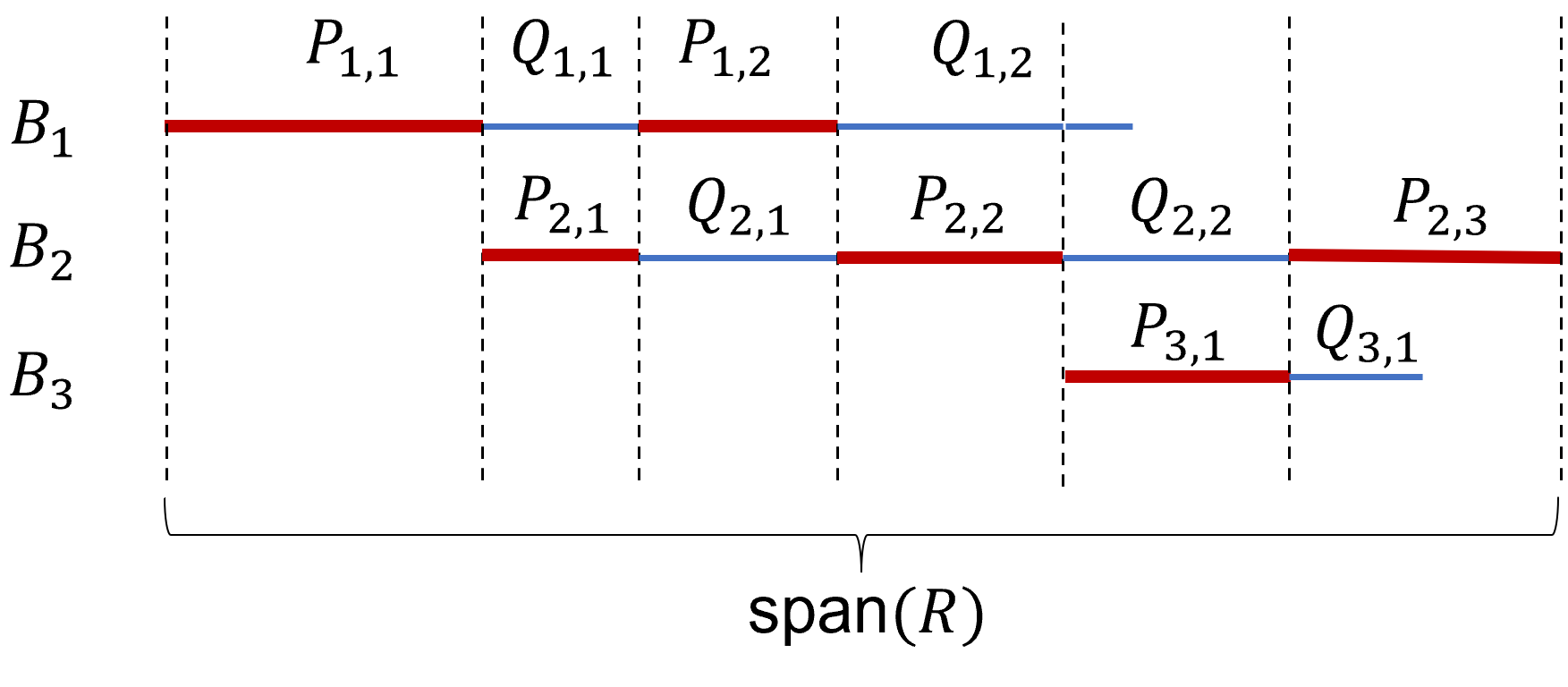}
\caption{\normalfont Shows the usage periods of 3 bins used by Move To Front decomposed into leading (red/thick intervals) and non-leading intervals (blue/thin intervals). The span is also indicated.}
\label{fig:mtf}
\end{figure}

\noindent We analyze the two summands of \eqref{eq:mtf-cost} separately. First we show:

\begin{claim}\label{clm:P-int}
$\displaystyle \sum_{i=1}^m \sum_{j=1}^{n_i} \ell(P_{i,j}) \le \opt(\R).$
\end{claim}
\begin{proof}
At each time $t$, exactly one bin is the leader, hence the leading intervals of bins $B_i$ and $B_{i'}$ are disjoint, i.e, $P_{i,j} \cap P_{i',j'} = \emptyset$ for any $i,i'\in [m]$,  $j\in[n_i]$, and $j'\in[n_{i'}]$. Since at each time $t\in[0,\span(\R))$ some bin is the leader, one can immediately observe that all the leading intervals partition the interval $[0,\span(\R))$ (see Figure~\ref{fig:mtf}). Combined with Lemma~\ref{lem:lb-opt} (iii), we arrive at Claim~\ref{clm:P-int}:
\[
\sum_{i=1}^m \sum_{j=1}^{n_i} \ell(P_{i,j}) = \span(\R) \le \opt(\R). \qedhere
\]
\end{proof}

We now analyze $\sum_{i=1}^m \sum_{j=1}^{n_i} \ell(Q_{i,j})$. For some $i\in[m]$ and $j\in[n_i]$, consider a non-leading interval $Q_{i,j}$ beginning at time $t_{i,j}$, which is preceded by a leading interval $P_{i,j}$ which ends at $t_{i,j}$. The reason that bin $B_i$ ceased to be a leader at time $t_{i,j}$ is because some other bin $B_{i'}$ received a new item $r_{i,j}$ and became the leader at time $t_{i,j}$. Thus, the algorithm was unable to pack item $r_{i,j}$ in bin $B_i$, the previous leader. Let $R_{i,j} \subseteq R_i$ be the set of items active in bin $i$ at the start of the interval $Q_{i,j}$, i.e. at time $t_{i,j}$. Then it must be mean that for some dimension $k\in[d]$, $(\size(r_{i,j})+\size(R_{i,j}))_k > 1$, or equivalently $\norm{\size(r_{i,j}) + \size(R_{i,j})}_\infty > 1$. Together with Proposition~\ref{prop:norm}, we obtain:

\begin{equation}\label{eq:mtf-rR}
\begin{aligned}
&\sum_{i=1}^m \sum_{j=1}^{n_i} \ell(Q_{i,j}) < \sum_{i=1}^m \sum_{j=1}^{n_i} \norm{\size(r_{i,j}) + \size(R_{i,j})}_\infty \cdot \ell(Q_{i,j})\\
&\le \sum_{i=1}^m \sum_{j=1}^{n_i} \norm{\size(r_{i,j})}_\infty\cdot \ell(Q_{i,j}) + \sum_{i=1}^m \sum_{j=1}^{n_i}\norm{\size(R_{i,j})}_\infty\cdot \ell(Q_{i,j}),
\end{aligned}
\end{equation}
We analyze the two summands of \eqref{eq:mtf-rR} separately. First we show:
\begin{claim}\label{clm:Q-int1}
$\displaystyle \sum_{i=1}^m \sum_{j=1}^{n_i} \norm{\size(r_{i,j})}_\infty\cdot \ell(Q_{i,j}) \le \mu\cdot d \cdot \opt(\R).$
\end{claim}
\begin{proof}
Observe that since no new item is packed in a bin $B_i$ during a non-leading interval $Q_{i,j}$, we have $\ell(Q_{i,j}) \le \mu$, since each item has a duration of at most $\mu$. Moreover, the items $r_{i,j}$ are distinct, since each $r_{i,j}$ is uniquely associated with the interval $Q_{i,j}$. Using these observations, we obtain the claim as follows:
\begin{equation*}\label{eq:mtf-r}
\begin{aligned}
&\sum_{i=1}^m \sum_{j=1}^{n_i} \norm{\size(r_{i,j})}_\infty \cdot \ell(Q_{i,j}) \le \sum_{i=1}^m \sum_{j=1}^{n_i} \norm{\size(r_{i,j})}_\infty \cdot \mu \\
&\le \mu \cdot \bigg(\sum_{i=1}^m \sum_{j=1}^{n_i} \norm{\size(r_{i,j})}_\infty \cdot \ell(I(r_{i,j}))\bigg) \quad\:(\text{since } \ell(I(r)) \ge 1) \\
&\le \mu\cdot \bigg(\sum_{r\in\R} \norm{\size(r)}_\infty\cdot \ell(I(r))\bigg) \\
&\le \mu\cdot d \cdot \opt(\R). \quad\qquad\qquad\qquad\qquad \text{(using Lem.~\ref{lem:lb-opt} (ii))} 
\end{aligned} 
\end{equation*}
\end{proof}

\noindent We next analyze the second summand of \eqref{eq:mtf-rR}.

\begin{claim}\label{clm:Q-int2}
$\displaystyle \sum_{i=1}^m \sum_{j=1}^{n_i} \norm{\size(R_{i,j})}_\infty \cdot \ell(Q_{i,j}) \le (\mu+1)\cdot d\cdot \opt(\R)$.
\end{claim}
\begin{proof}
Using Proposition~\ref{prop:norm}, observe the following:
\begin{equation}\label{eq:mtf-R}
\begin{aligned}
&\sum_{i=1}^m \sum_{j=1}^{n_i} \norm{\size(R_{i,j})}_\infty \cdot \ell(Q_{i,j}) \le \sum_{i=1}^m \sum_{j=1}^{n_i} \sum_{r\in R_{i,j}}\norm{\size(r)}_\infty \cdot \ell(Q_{i,j}) \\
&= \sum_{i=1}^m \sum_{r\in R_i} \norm{\size(r)}_\infty \cdot \bigg(\sum_{j\in[n_i]: r\in R_{i,j}} \ell(Q_{i,j})\bigg),
\end{aligned}
\end{equation}
where the last equality follows by changing the order of summation.

Consider an item $r\in R_i$, and let $j^-, j^+ \in [n_i]$ be such that $r\in R_{i,j}$ for all $j\in[j^-,j^+]$, i.e., item $r$ is active during the intervals $Q_{i,j^-},\dots,Q_{i,j^+}$. Since $r$ is active during the start of each interval $Q_{i,j}$ for each $j\in[j^-,j^+]$, we have $\cup_{j=j^-}^{j^+-1} \ell(Q_{i,j}) \subseteq I(r)$. This implies: 
\begin{equation}\label{eq:mtf-clm3}
\sum_{j\in[n_i]: r\in R_{i,j}} \ell(Q_{i,j}) = \sum_{j=j^-}^{j^+-1} \ell(Q_{i,j}) + \ell(Q_{i,j^+}) \le \ell(I(r)) + \mu, 
\end{equation}
where we used $\ell(Q_{i,j^+})\le \mu$. Using the above in eq.~\eqref{eq:mtf-R}, we obtain:
\begin{equation*}\label{eq:mtf-R2}
\begin{aligned}
&\sum_{i=1}^m \sum_{j=1}^{n_i} \norm{\size(R_{i,j})}_\infty \cdot \ell(Q_{i,j}) \\
&\le \sum_{i=1}^m \sum_{r\in R_i} \norm{\size(r)}_\infty \cdot \bigg(\sum_{j\in[n_i]: r\in R_{i,j}} \ell(Q_{i,j})\bigg) &\text{(from \eqref{eq:mtf-R})}\\
&\le \sum_{i=1}^m \sum_{r\in R_i} \norm{\size(r)}_\infty \cdot (\ell(I(r)) + \mu) &\text{(from \eqref{eq:mtf-clm3})}\\
&\le (\mu+1)\cdot\bigg(\sum_{r\in \R} \norm{\size(r)}_\infty \cdot \ell(I(r)))\bigg) &\text{(since $\ell(I(r))\ge 1$)}\\
&\le (\mu+1)\cdot d\cdot \opt(\R), & \text{(using Lem.~\ref{lem:lb-opt} (ii))}
\end{aligned}
\end{equation*}
thus proving Claim~\ref{clm:Q-int2}.
\end{proof}
\noindent Claims~\ref{clm:P-int},\ref{clm:Q-int1} and \ref{clm:Q-int2} together with equations~\eqref{eq:mtf-cost} and \eqref{eq:mtf-rR} imply:
\[\cost(MF,\R) \le ((2\mu+1) d + 1)\cdot \opt(\R), \]
thus proving Theorem~\ref{thm:ub-mtf}.

\section{Upper Bound on the Competitive Ratio of First Fit}\label{sec:ff}
In this section, we prove an upper bound on the CR of First Fit.
\begin{theorem}\label{thm:ub-ff}
The competitive ratio of First Fit for the MinUsageTime Dynamic Vector Bin packing problem is at most $(\mu+2)d+1$.
\end{theorem}

Let $B_1,\dots, B_m$ be the bins used by First Fit (FF) to pack an item sequence $\R$. For $i\in[m]$, let $R_i\subseteq \R$ be the items packed in bin $B_i$, and let $I_i = [I_i^-,I_i^+)$ denote the active interval of $B_i$. We assume bins are indexed according to their opening times, i.e., $I_1^- \le \dots \le I_m^-$. 

\begin{figure}[t]
\centering
\includegraphics[scale=0.5]{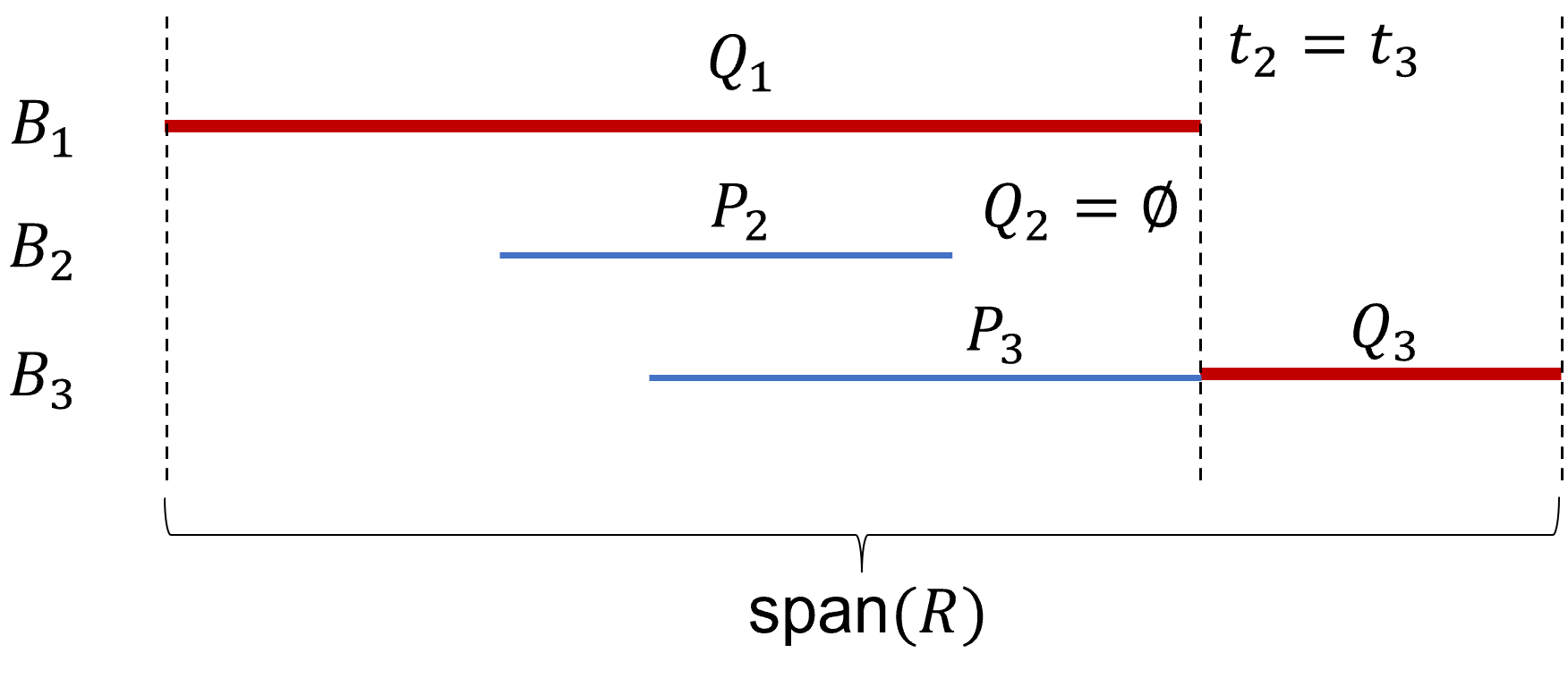}
\caption{\normalfont Shows the decomposition of the usage periods of 3 bins used by First Fit. The span is also indicated.}
\label{fig:ff}
\end{figure}

The cost of the First Fit packing is given by $\cost(FF, \R) = \sum_{i=1}^m \ell(I_i)$. Following the ideas of \cite{RenTangLiCai2017}, we decompose each interval $I_i$ as follows. Let $t_i$ denote the latest closing time of bins opened before $B_i$, i.e., $t_i = \max(I_i^-, \max_{j<i} I_j^+)$. Then we partition $I_i = P_i \cup Q_i$, where $P_i = [I_i^-, \min(I_i^+, t_i))$ and $Q_i = [\min(I_i^+, t_i), I_i^+)$. Note that $P_1 = \emptyset$. The decomposition is illustrated in Figure~\ref{fig:ff}. Therefore, the cost of the packing is:

\begin{equation}\label{eq:ff-cost-0}
\cost(FF, \R) = \sum_{i=2}^m \ell(P_i) + \sum_{i=1}^m \ell(Q_i).
\end{equation}

\begin{claim}\label{clm:ff-0}
$\sum_{i=1}^m \ell(Q_i) = \span(\R) \le \opt(\R).$
\end{claim}
\begin{proof}
The claim follows directly from the definition of the decomposition (see Fig.~\ref{fig:ff}) and Lemma~\ref{lem:lb-opt} (iii).
\end{proof}

Let us now define $R'_i \subseteq R_i$ to be an inclusion-wise minimal cover of the interval $P_i$. That is, the union of active intervals of items in $R'_i$ covers $P_i$, but any $J\subset R'_i$ does not cover $P_i$. Let $r_{i,1}, \dots, r_{i,n_i}$ be the $n_i$ items in $R'_i$, sorted by their arrival time. By the minimality of $R'_i$, each item in $R'_i$ has a distinct arrival time. Thus we can index the items so that $a(r_{i,1}) < a(r_{i,2}) < \dots <  a(r_{i,n_i})$.
Moreover, the minimality of $R'_i$ also implies that ending times of the items are in sorted order, i.e., $e(r_{i,1}) < e(r_{i,2}) < \dots <  e(r_{i,n_i})$; if not, an item can be removed from $R'_i$ while still covering $P_i$, thus contradicting the minimality of $R'_i$. 

We now decompose each non-empty interval $P_i$ into $n_i$ disjoint periods $P_i = P_{i,1} \cup \dots \cup P_{i,n_i}$ where $P_{i,j} = [a(r_{i,j}), a(r_{i,j+1}))$ for $1\le j < n_i$ and $P_{i,n_i} = [a(r_{i,n_i}), \min(I_i^+,t_i))$. Since this is a partition of $P_i$, we have $\ell(P_i) = \sum_{j=1}^{n_i} \ell(P_{i,j})$ for each $i\ge 2$.

For an item $r_{i,j} \in R'_i$, we refer to the largest index bin with index less than $i$ which is open at time $a(r_{i,j})$ as the \textit{blocking bin}\footnote{\cite{RenTangLiCai2017} use the terminology supplier bin instead} $B(i,j)$ for the item $r_{i,j}$ and the interval $P_{i,j}$. Note that since an item $r_{i,j}$ is placed in bin $B_i$, all previously opened bins including the blocking bin $B(i,j)$ could not pack $r_{i,j}$ when it arrived. Thus:
\begin{equation*}
    \norm{\size(r_{i,j}) + \size(R_{i,j})}_\infty > 1,
\end{equation*}
where $R_{i,j}$ is the set of items in $B(i,j)$ that are active at time $a(r_{i,j})$.
Using this, we have:
\begin{equation}\label{eq:ff-cost-2}
\begin{aligned}
&\sum_{i=2}^m \sum_{j=1}^{n_i} \ell(P_{i,j}) < \sum_{i=2}^m \sum_{j=1}^{n_i} \norm{\size(r_{i,j}) + \size(R_{i,j})}_\infty \cdot \ell(P_{i,j})\\
&\le \sum_{i=2}^m \sum_{j=1}^{n_i} \norm{\size(r_{i,j})}_\infty\cdot \ell(P_{i,j}) + \sum_{i=2}^m \sum_{j=1}^{n_i}\norm{\size(R_{i,j})}_\infty\cdot \ell(P_{i,j}),
\end{aligned}
\end{equation}

\noindent We analyze the summands of \eqref{eq:ff-cost-2} separately. We first have:
\begin{claim}\label{clm:ff-1}
$\sum_{i=2}^m \sum_{j=1}^{n_i} \norm{\size(r_{i,j})}_\infty\cdot \ell(P_{i,j}) \le d\cdot \opt(\R).$
\end{claim}
\begin{proof}
By definition of $P_{i,j}$, we have $P_{i,j}\subseteq I(r_{i,j})$. Thus, $\ell(P_{i,j}) \le \ell(I(r_{i,j}))$. Lemma \ref{lem:lb-opt} (ii) then proves the claim.
\end{proof}

\noindent The next claim analyzes the second summand of \eqref{eq:ff-cost-2}.
\begin{restatable}{claim}{clmff}\label{clm:ff-2}
$\sum_{i=2}^m \sum_{j=1}^{n_i} \norm{\size(R_{i,j})}_\infty\cdot \ell(P_{i,j}) \le (\mu+1)\cdot d\cdot\opt(\R).$
\end{restatable}
\begin{proof}
Let $\hat{R} = \cup_{i=2}^m\cup_{j=1}^{n_i} R_{i,j}$ be the set of all items belonging to bins considered as blocking bins by items in $\{R'_i\}_{i\ge 2}$. We have:
\begin{equation}\label{eq:ff-cost-3}
\begin{aligned}
&\sum_{i=2}^m \sum_{j=1}^{n_i} \norm{\size(R_{i,j})}_\infty \cdot \ell(P_{i,j}) \le \sum_{i=2}^m \sum_{j=1}^{n_i} \sum_{r\in R_{i,j}}\norm{\size(r)}_\infty \cdot \ell(P_{i,j}) \\
&= \sum_{r\in \hat{R}} \norm{\size(r)}_\infty \cdot \bigg(\sum_{(i,j): r\in R_{i,j}} \ell(P_{i,j})\bigg),
\end{aligned}
\end{equation}
where the last inequality follows by changing the order of summation. Now for a fixed $r\in \hat{R}$ which is packed in some bin $B$, consider two distinct items $r_{i,j}$ and $r_{i',j'}$ s.t. $r\in R_{i,j} \cap R_{i',j'}$. We will show that $P_{i,j} \cap P_{i',j'} = \emptyset$. \begin{itemize}[leftmargin=*]
\item For $i=i'$, this follows from the fact that $\{P_{i,j}\}_{j=1}^{n_i}$ partitions $P_i$.
\item For $i\neq i'$, let $i<i'$ w.l.o.g. Then since $r_{i,j}$ and $r_{i',j'}$ have the same blocking bin $B$, it must be the case that when $r_{i',j'}$ arrives, $B_i$ must be closed, otherwise $B_i$ would be the blocking bin for $r_{i',j'}$. Thus, $r_{i,j}$ must have departed when $r_{i',j'}$ arrives, implying that $P_{i,j} \cap P_{i',j'} = \emptyset$.
\end{itemize}

\noindent Thus for a given $r\in\hat{R}$, the set of intervals $P_{i,j}$ s.t. $r\in R_{i,j}$ are pairwise disjoint. Hence we can observe that for each $r\in \hat{R}$:
\begin{equation}\label{eq:ff-time}
\sum_{(i,j): r\in R_{i,j}} \ell(P_{i,j}) \le \max_{(i,j): r\in R_{i,j}} e(r_{i,j})- \min_{(i,j): r\in R_{i,j}} a(r_{i,j}).  
\end{equation}

\noindent Note that since each $r\in\hat{R}$ is active at the arrival time of an item $r_{i,j}$ s.t. $r\in R_{i,j}$, we have $a(r) \le a(r_{i,j}) \le e(r)$. Thus, $e(r_{i,j}) \le \mu + a(r_{i,j}) \le \mu + e(r)$. Putting these in \eqref{eq:ff-time}, we obtain:
\begin{equation*}\label{eq:ff-time-2}
\sum_{(i,j): r\in R_{i,j}} \ell(P_{i,j}) \le \mu 
+ e(r) - a(r) \le \mu + \ell(I(r)) \le (\mu+1)\cdot \ell(I(r)).  
\end{equation*}
Using the above in \eqref{eq:ff-cost-3} with Lemma~\ref{lem:lb-opt} (ii), we see that
\begin{equation*}
\begin{aligned}
&\sum_{i=2}^m \sum_{j=1}^{n_i} \norm{\size(R_{i,j})}_\infty \cdot \ell(P_{i,j}) \le \sum_{r\in \hat{R}} \norm{\size(r)}_\infty \cdot \bigg(\sum_{(i,j): r\in R_{i,j}} \ell(P_{i,j})\bigg) \\
&\le (\mu+1)\cdot\sum_{r\in \hat{R}} \norm{\size(r)}_\infty \cdot \ell(I(r)) \le (\mu+1)\cdot d\cdot \opt(\R),
\end{aligned}
\end{equation*}
thus proving the claim.
\end{proof}

\noindent Claims~\ref{clm:ff-0}, \ref{clm:ff-1} and \ref{clm:ff-2} together with equations \eqref{eq:ff-cost-0} and \eqref{eq:ff-cost-2} imply:
\[\cost(FF,\R) \le ((\mu+2)d+1)\cdot\opt(\R),\]
thus proving Theorem~\ref{thm:ub-ff}.

\section{Upper Bound on the Competitive Ratio of Next Fit}\label{sec:nf}
In this section, we prove an upper bound on the CR of Next Fit.
\begin{theorem}\label{thm:ub-nextfit}
The competitive ratio of Next Fit for the MinUsageTime Dynamic Vector Bin packing problem is at most $2\mu d+1$. 
\end{theorem}

Let $B_1,\dots, B_m$ be the bins used by Next Fit (NF) to pack an item sequence $\R$. As before, for $i\in[m]$, let $R_i\subseteq \R$ be the items packed in bin $B_i$, and let $I_i$ denote the active interval of $B_i$. We have $\cost(NF, \R) = \sum_{i=1}^m \ell(I_i)$.

Recall that Next Fit maintains one current bin at a time into which it tries to pack incoming items. Following~\cite{KamaliLopezOrtiz2015}, we decompose the usage period $I_i$ of a bin $B_i$ into two intervals $P_i$ and $Q_i$ based on when Next Fit considered $B_i$ as the current bin. We decompose the interval $I_i = [I_i^-, I_i^+)$ as
$I_i = P_i \cup Q_i$, where $P_i = [I_i^-, t_i)$ and $Q_i = [t_i, I_i^+)$, with $t_i\in I_i$ denoting the time at which $B_i$ was released. Thus, $P_i$ is the time period when $B_i$ was considered the current bin and $Q_i$ is the time period when $B_i$ ceased to the current bin. 

Using the above interval-decomposition, we can write the cost as $\cost(NF,\R) = \sum_{i=1}^m \ell(P_i)+ \sum_{i=1}^m\ell(Q_i)$. Note that at each time $t$, exactly one bin is current, hence $P_i \cap P_{i'} = \emptyset$ for all $i\neq i'$. Further at each time some bin is current, hence we conclude that the intervals $\{P_i\}_{i\in[m]}$ partition the interval $[0,\span(\R)]$. Together with Lemma~\ref{lem:lb-opt} (iii), this gives:
\begin{equation}\label{eq:nf-span}
\sum_{i=1}^m \ell(P_i) = \span(\R) \le \opt(\R).
\end{equation}

Next, observe that at a bin $B_i$ was released at time $t_i$ because an item $r_i$ could not be packed into $B_i$. This means: 
\begin{equation}\label{eq:nf-size}
    \norm{\size(R'_i) + \size(r_i)}_\infty > 1,
\end{equation}
where $R'_i\subseteq R_i$ denotes the items packed in $B_i$ which are active at $t_i$. Moreover, since a bin $B_i$ is released at $t_i$, it does not receive any new item in the period $Q_i$. Thus $\ell(Q_i) \le \mu$, for each $i\in[m]$. We use these observations to prove the following.

\begin{equation*}
\begin{aligned}
&\sum_{i=1}^m \ell(Q_i) < \sum_{i=1}^m \norm{\size(R'_i)+\size(r_i)}_\infty \cdot \ell(Q_i) & \text{(using \eqref{eq:nf-size})}\\
&\le \mu\cdot\sum_{i=1}^m \norm{\size(R'_i)+\size(r_i)}_\infty & \text{(since $\ell(Q_i) \le \mu$)}\\
&\le \mu\cdot\sum_{i=1}^m\sum_{r\in R'_i}\norm{\size(r)}_\infty+ \mu\cdot\sum_{i=1}^m\norm{\size(r_i)}_\infty & \text{(using Prop.~\ref{prop:norm})}\\
&\le 2\mu\cdot \sum_{r\in \R} \norm{\size(r)}_\infty \le 2\mu\cdot d\cdot\opt(\R), & \text{(using Lemma~\ref{lem:lb-opt} (ii))}
\end{aligned}
\end{equation*}
where the last transition uses the observation that each item $r\in \R$ can appear in the summation at most twice: once in some $R'_i$ and once as some $r_i$. Together with \eqref{eq:nf-span}, we conclude that $\cost(NF,\R) \le (2\mu d+1)\cdot\opt(\R)$, thus proving Theorem~\ref{thm:ub-nextfit}.

\section{Lower Bounds on the Competitive Ratio of Any Fit Packing Algorithms}\label{sec:lb}

In this section, we prove lower bounds on the competitive ratio of algorithms belonging to the Any Fit packing algorithm family. Our first result applies to all algorithms in this family.

\begin{theorem}\label{thm:lb-anyfit}
The competitive ratio of any Any Fit packing algorithm for the MinUsageTime Dynamic Vector Bin Packing problem in $d$-dimensions is at least $(\mu+1)d$.
\end{theorem}
\begin{proof}
We present the following worst-case adversarial example against which any Any Fit algorithm $\A$ has competitive ratio at least $(\mu+1)d$. Let $k\ge 1$ be a parameter. Let $\eps, \eps' \in (0,1)$ be such that $\eps>\eps'$, $d^2\eps k < 1$, $d\eps > 2\eps'$, and $\eps(1+d)< 1$. 

We first construct a sequence of $n = 2\cdot d\cdot k$ items labelled as $\R_0 = \{1,2,\dots,2dk-1,2dk\}$. To specify the size of the items, we partition them into groups $G_0, G_1, G_2, \dots,$ and $G_d$. Here $G_0 = \{j\in [2dk]: j \text{ is even}\}$ is the set of even-indexed items. The size of an item $j$ in $G_0$ is the following vector:
\[\size(j) = \begin{bmatrix}
d\eps - \eps' && d\eps - \eps' && \cdots && d\eps - \eps'
\end{bmatrix},\]
i.e., a vector equal to $(d\eps - \eps')\cdot \vect{1}^d$.

For $i\in[d]$, the group $G_i$ is the set $\{2m -1: (i-1)\cdot k + 1\le m \le i\cdot k\}$ of $k$ items, i.e., odd-indexed items in the range $[2(i-1)k+1, 2ik]$. The size of an item $j\in G_i$ is the vector with $(1-d\eps)$ in the $i^{th}$ dimension and $\eps$ everywhere else, i.e.,
\[\size(j) = \begin{bmatrix}
\eps && \cdots && (1-d\eps) &&\cdots && \eps
\end{bmatrix}.\]

The items $\R_0 = \{1,2,\dots,2dk\}$ arrive in that order at time 0, and their active interval is $[0,1)$. Consider the execution of an Any Fit packing algorithm $\A$. Items $1$ and $2$ are initially placed into a single bin $B_1$ after which the bin is loaded at $(1-d\eps + d\eps - \eps') = 1-\eps'$ in dimension 1. Now no item $j$ for $j\ge 3$ cannot be packed into $B_1$ since the load in dimension 1 would exceed the capacity as we have $1-\eps' + \eps > 1$ and $1-\eps' + (d\eps-\eps') > 1$. Hence another bin $B_2$ is opened. Continuing in this manner, one can observe that at least $dk$ bins $B_1,B_2,\dots,B_{dk}$ are created, with bins $B_{(i-1)k+1}, \dots, B_{ik}$ being loaded at level $(1-\eps')$ in dimension $i$ and at $(\eps+d\eps-\eps')$ in other dimensions, for $i\in [d]$. Thus, $\cost(\A,\R_0) \ge dk$.

On the other hand, $\opt(\R_0) \le k+1$. This is because all the even-indexed items of $G_0$ can be packed into one bin $B_0$ since $(d\eps-\eps')\cdot (dk) < 1$. The remaining items can be packed into $k$ bins, each of which contains exactly one item from $G_i$, for each $i\in [d]$. This is a feasible packing since the load on the $j^{th}$ dimension of any such bin is $(1-d\eps) + (d-1)\cdot\eps = 1-\eps < 1$.

We now introduce a sequence $\R_1$ of $dk$ identical items, each of which are loaded at $\eps'$ in each dimension. These items arrive just before any items of $\R_0$ depart and their active interval is $[1,\mu+1)$. Consider the execution of $\A$ on $\R_0 \cup \R_1$. As argued earlier, $\A$ opens at least $dk$ bins which are loaded at at most $(1-\eps')$ in each dimension (since $(1+\eps)d < 1$, and exactly at $(1-\eps')$ in one dimension. Thus, each item of $\R_1$ will be packed in a separate bin by $\A$. This is because $\A$ is an Any Fit packing algorithm and will not open a new bin since the $dk$ items of $\R_1$ can be packed in the $dk$ bins created by $\A$ while packing $\R_0$. Subsequently, items in $\R_0$ depart, and each of the $dk$ bins contain one item each from $\R_1$ in the period $[1,\mu+1)$. Thus we have $\cost(\A,\R_0\cup\R_1) \ge dk(1+\mu)$.

On the other hand, the optimal algorithm can pack all items of $\R_1$ into the bin $B_0$ which held all even-indexed items, since $(d\eps-\eps')dk + dk\cdot\eps' = d^2\eps k < 1$. Thus only bin $B_0$ has an usage period of length $\mu+1$ while the remaining $k$ bins have a usage period of length $1$ since they only contain items from $\R_0$. Thus, $\opt(\R_0\cup\R_1) \le k+1+\mu$. Thus the competitive ratio of $\A$ is:
\[CR(\A) \ge \frac{\cost(\A,\R_0\cup\R_1)}{\opt(\R_0\cup\R_1)} \ge \frac{dk(\mu+1)}{k+\mu+1} = \frac{(\mu+1)d}{1+(\mu+1)/k}\]

Since $k$ is an arbitrary parameter, in the limit $k\rightarrow\infty$, we have $CR(\A) \ge (\mu+1)d$ for any Any Fit packing algorithm $\A$, proving the claimed lower bound.
\end{proof}

\begin{figure}
\centering
\includegraphics[scale=0.6]{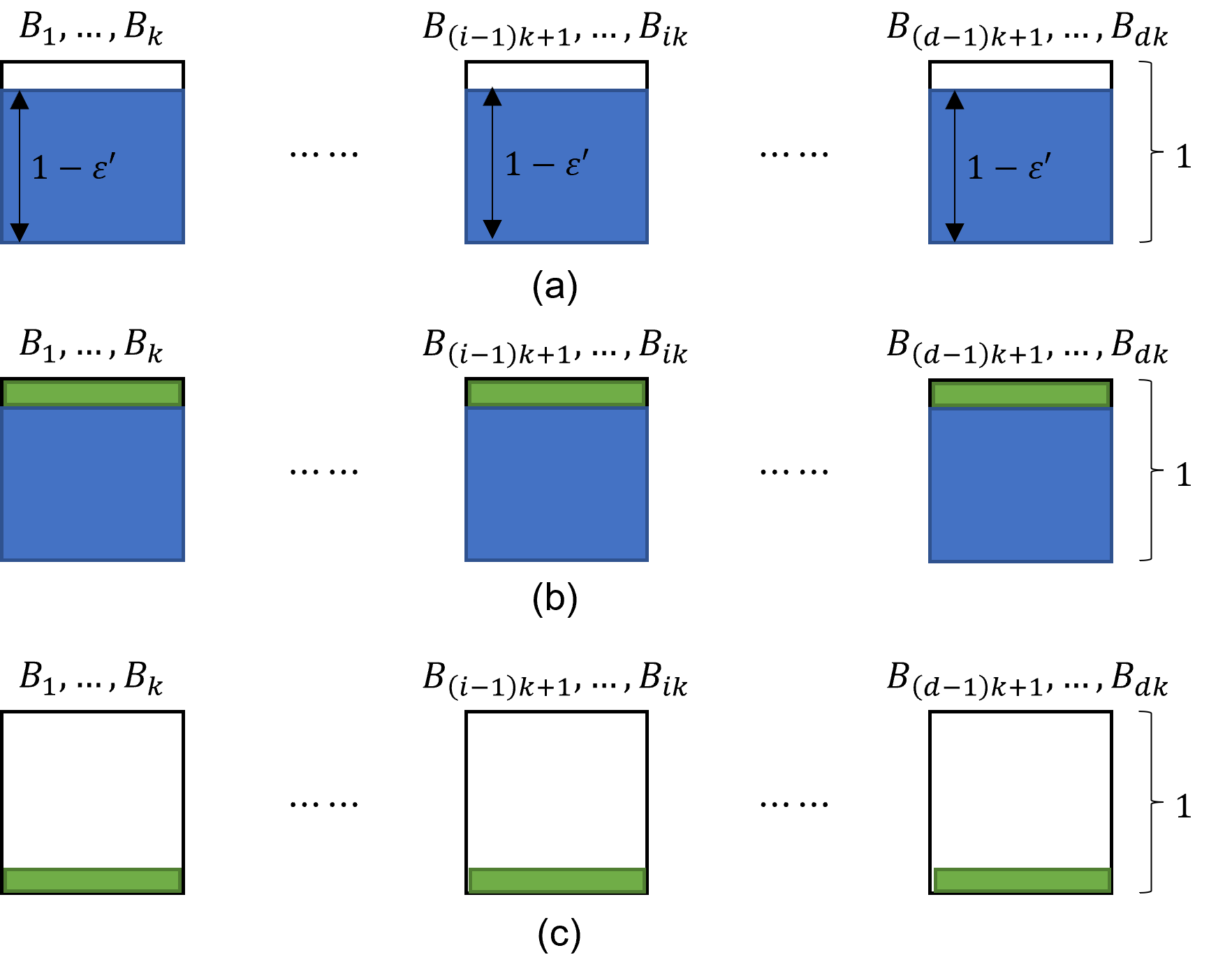}
\caption{\normalfont Illustrates the load on bins used by any Any Fit packing algorithm $\A$ on the item list $\R_0 \cup \R_1$. Part (a) shows $\A$ opens $dk$ bins in the time period $[0,1)$ where bins $B_{(i-1)k+1},\dots,B_{ik}$ have load $1-\eps'$ in dimension $i$. Part (b) shows that $\A$ packs $dk$ items of $\R_1$ at time $1$, fully loading each bin $B_{(i-1)k+1},\dots,B_{ik}$ in dimension $i$. Part (c) shows the time period $[1,\mu+1)$ when items of $\R_0$ have departed and each bin contains one item from $\R_1$.}
\label{fig:lb}
\end{figure}

The execution of any Any Fit packing algorithm $\A$ on the item list $\R_0\cup \R_1$ in illustrated in Figure~\ref{fig:lb}. We now prove a stronger lower bound against Next Fit using a different construction.
\begin{theorem}\label{thm:lb-nextfit}
The competitive ratio of Next Fit for the MinUsageTime Dynamic Vector Bin Packing problem is at least $2\mu d$.
\end{theorem}
\begin{proof}
We present the following worst-case adversarial example against which Next Fit has a competitive ratio at least $2\mu d$. Let $k\ge 2$ be an even integer. Let $\eps,\eps' \in (0,1)$ be such that $\eps'>2d\eps$ and $\eps'dk < 1$. We construct a sequence of $n = 2\cdot d \cdot k$ items labelled as $\R = \{1,2,\dots, 2dk\}$. As before, we partition the items into groups $G_0, G_1,\dots, G_d$, where $G_0 = \{j\in[2dk]: j \text{ is even}\}$ is the set of even-indexed items. The size of an item $j\in G_0$ is $\eps'\cdot\vect{1}^d$, i.e., 
\[\size(j) = \begin{bmatrix}
\eps' && \eps' && \cdots && \eps'
\end{bmatrix}. \]

For $i\in [d]$, the group $G_i$ is the set $\{2m -1: (i-1)\cdot k + 1\le m \le i\cdot k\}$ of $k$ items, i.e., odd-indexed items in the range $[2(i-1)k+1, 2ik]$. The size of an item $j\in G_i$ is the vector with $(1-d\eps)$ in the $i^{th}$ dimension and $\eps$ everywhere else, i.e.,
\[\size(j) = \begin{bmatrix}
\eps && \cdots && (\frac{1}{2}-d\eps) &&\cdots && \eps
\end{bmatrix}.\]
Items $\R= \{1,2,\dots, 2dk\}$ arrive in that order at time 0. The active interval of items in $G_0$ is $[0,\mu)$ and of those in $\cup_{i=1}^d G_i$ is $[0,1)$. 

Consider the execution of Next Fit (NF) on $\R$. Items $1$ and $2$ are placed in one bin $B_1$, which is then loaded at $1/2-d\eps + \eps'$. Item $3$ does not fit in $B_1$ since the load on dimension 1 would then be $2(1/2-d\eps) + \eps' > 1$. Thus NF closes $B_1$ and opens a bin $B_2$ to accommodate item $3$, after which item $4$ is also placed in $B_2$. Continuing this way, one observes that NF creates $k$ bins to pack the items $\{1,2,\dots,2k\}$, with the last bin $B_k$ being loaded at $(1/2-d\eps + \eps')$ in dimension 1 and $(\eps+\eps')$ in other dimensions. Subsequently, items $(2k+1)$ which is loaded at $(1/2-d\eps)$ in dimension 2, and item$(2k+2)$ can also be placed in $B_k$, since $(\eps+\eps'+1/2-d\eps + \eps') < 1$. This loads the bin $B_k$ at $(1/2-d\eps+\eps+2\eps')$ in dimension 2. 
Consequently, item $(2k+3)$ cannot be placed in this bin since it has load $(1/2-d\eps)$ in dimension 2. 
Thus NF closes $B_k$ and opens another bin. It can be seen that for the rest of the items in $\{2k+3,\dots,4k\}$, NF opens $(k-1)$ bins, similar to the execution on items $\{1,\dots, 2k\}$. Continuing this way for $(d-1)$ more phases, with phase $i$ corresponding to packing of the items in $\{2(i-1)k+1, \dots,2ik\}$ for $2\le i\le d$, one can observe that NF opens $k + (k-1) + \cdots + (k-1) = 1+(k-1)d$ bins. Since each bin contains at least one even-indexed item, each bin is active for a duration of $\mu$. Thus, $\cost(NF,\R) \ge (1+(k-1)d)\mu$.

On the other hand, $\opt(\R) \le \mu+k/2$. This is because the optimal algorithm can pack all items of $G_0$ into a single bin, since $\eps'\cdot dk < 1$. This bin will be active in the interval $[0,\mu)$, leading to a cost of $\mu$. The remaining items can be packed into $k/2$ bins, each of which contains exactly two items from $G_i$, for each $i\in [d]$. This is a feasible packing since the load on the $j^{th}$ dimension of any such bin is $2\cdot(1/2-d\eps) + (d-1)\cdot 2\cdot\eps = 1-2\eps < 1$. Each such bin is active for the interval $[0,1)$, contributing a cost of 1. Thus the competitive ratio of Next Fit satisfies:
\[CR(NF) \ge \frac{\cost(NF,\R)}{\opt(\R)} \ge \frac{(1+(k-1)d)\mu}{\mu + k/2} = \frac{2\mu(d + \frac{1}{k-1})}{\frac{k}{k-1} + \frac{2\mu}{k-1}}.\]
Since $k$ is arbitrary, in the limit $k\rightarrow\infty$, we have $CR(NF) \ge 2\mu d$, proving the claimed lower bound.
\end{proof}

Note that our lower bounds (Theorems \ref{thm:lb-anyfit} and \ref{thm:lb-nextfit}) match the lower bound of $(\mu+1)$ for Any Fit packing algorithms \cite{LiTangCai2016,RenTangLiCai2017} and $2\mu$ for Next Fit \cite{TangLiRenCai2016} for the one-dimensional case. However the one-dimensional examples of \cite{LiTangCai2016,TangLiRenCai2016} do not generalize to multiple dimensions, and hence the new constructions of Theorems~\ref{thm:lb-anyfit} and \ref{thm:lb-nextfit} are needed. We also record that the competitive ratio of Best Fit can be unbounded, even for $d=1$, as was shown in \cite{LiTangCai2016}.
\begin{theorem}[\cite{LiTangCai2016}] The competitive ratio of Best Fit for the MinUsageTime Dynamic Vector Bin Packing problem is unbounded.
\end{theorem}

Lastly, we examine the competitive ratio of Move To Front.
\begin{restatable}{theorem}{lbmtf}\label{thm:lb-mtf}
The competitive ratio of the Move To Front algorithm for the MinUsageTime Dynamic Vector Bin Packing problem is at least $\max\{2\mu, (\mu+1)d\}$.
\end{restatable}
\begin{proof}
Note that the lower bound of $(\mu+1)d$ follows from Theorem~\ref{thm:lb-anyfit} since Move To Front is an Any Fit packing algorithm. For the lower bound of $2\mu$, consider the following one-dimensional example.

For a parameter $n\ge 1$, let $\R = \{1,2,\dots, 4n\}$ be a sequence of $4n$ items which arrive in that order at time $0$. The odd-indexed items have size $1/2$ and active interval $[0,1)$. The even-indexed items have size $1/(2n)$ and active interval $[0,\mu)$.

Consider the execution of Move To Front on $\R$. Items $1$ and $2$ are placed into a single bin $B_1$, which is then loaded at $1/2+1/(2n)$. Therefore item 3 does not fit into bin $B_1$, and a new bin $B_2$ is opened to accommodate item $3$. Since $B_2$ is now the most-recently used bin, it is moved ahead of $B_1$, and therefore receives the next item $4$, leading to a load of $1/2+1/(2n)$. Now item $5$ cannot be placed in either $B_1$ or $B_2$, causing another bin to be opened. Continuing this way, we can observe that Move To Front (MF) will create $2n$ bins. Since each bin contains an even-indexed item, each bin will be active for a duration of $\mu$. Thus $\cost(MF, \R) = 2n\cdot\mu$.

On the other hand, the optimum algorithm can pack all the even-indexed items into one bin since there are $2n$ such items of size $1/(2n)$. This bin is active in $[0,\mu)$. The remaining $2n$ odd-indexed items of size $1/2$ each can be paired up an placed in $n$ bins, each of which is active in $[0,1)$ This implies $\opt(\R) \le \mu+n$. Thus the competitive ratio of Move To Front satisfies:
\[CR(MF) \ge \frac{\cost(MF,\R)}{\opt(\R)} \ge \frac{2n\mu}{\mu + n} = \frac{2\mu}{1+\frac{\mu}{n}}.\]
Since $n$ is arbitrary, in the limit $n\rightarrow\infty$, we have $CR(NF) \ge 2\mu$, proving the claimed lower bound. 
We note that the same example was used by \cite{RenTangLiCai2017,TangLiRenCai2016} to establish a lower bound of $2\mu$ on the competitive ratio of Next Fit.
\end{proof}

\begin{remark}\normalfont
The above theorem implies that the lower bound is $2\mu$ for $d=1$, and $(\mu+1)d$ for $d\ge 2$. We leave improving the lower bound of Move To Front for $d\ge 2$ as an interesting open question.
\end{remark}

\section{Experimental Evaluation}\label{sec:experiments}
In this section, we perform an experimental study evaluating the average-case performance of several Any Fit packing algorithms. 

\paragraph{Experimental Setup.}
In addition to Move To Front, First Fit and Next Fit, we study the following Any Fit algorithms:
\begin{enumerate}[leftmargin=*]
\item Best Fit, with the load of a bin containing a set $R$ of items is defined as $w(R) = \norm{\size(R)}_\infty$.
\item Worst Fit, which tries to place an item in the least loaded bin.
\item Last Fit, which in contrast with First Fit, tries to place an item in the bin which with the latest opening time.
\item Random Fit, which tries to place an item in a bin selected uniformly at random from the list of open bins.
\end{enumerate}

We evaluate the performance of these algorithms on randomly-generated input sequences. Our experimental setup closely follows the setup of \cite{KamaliLopezOrtiz2015} for the 1-D case.  In $d$-dimensions, we assume that each bin has size $\vect{B}^d$, for integers $d, B\ge 1$. Each item is assumed to have a size in $\{1,2,\dots,B\}^d$. For an integral value $T$ of the span, we assume each item arrives at an integral time step in $[0,T-\mu]$ and has an integral duration in $[1,\mu]$, for integral $\mu\ge 1$. Each instance in our experiment is a sequence of $n$ items, where the size and duration of each item is selected randomly from their ranges, assuming a uniform distribution. For different settings of the parameters as described in Table~\ref{tab:expts}, we generate $m = 1000$ input instances. Since the computation of the optimal packing is NP-hard, we evaluate the performance of an algorithm by comparing its packing cost to the lower bound on OPT from Lemma~\ref{lem:lb-opt} (i).

\begin{table}[b]
\centering
\begin{tabular}{|c|c|c|}\hline
Parameter & Description & Value \\\hhline{|=|=|=|}
$d$ & Num. dimensions & $\{1,2,5\}$ \\
$n$ & Sequence length & $n=1000$ \\
$\mu$ & Max. item length & $\{1,2,5,10,100,200\}$ \\
$T$ & Sequence span & $T=1000$ \\
$B$ & Bin size & $B=100$ \\\hline
\end{tabular}
\caption{\normalfont Summary of experimental parameters}
\label{tab:expts}
\end{table}

Our experimental results are shown in Figure~\ref{fig:expts}. 
For each combination of parameters $d\in\{1,2,5\}$ and $\mu\in\{1,2,5,10,100,200\}$, we plot the average performance of our algorithms in consideration, with error bars measuring the standard deviation. 

\begin{figure}
\centering
\begin{subfigure}[b]{0.47\textwidth}
   \includegraphics[width=\textwidth]{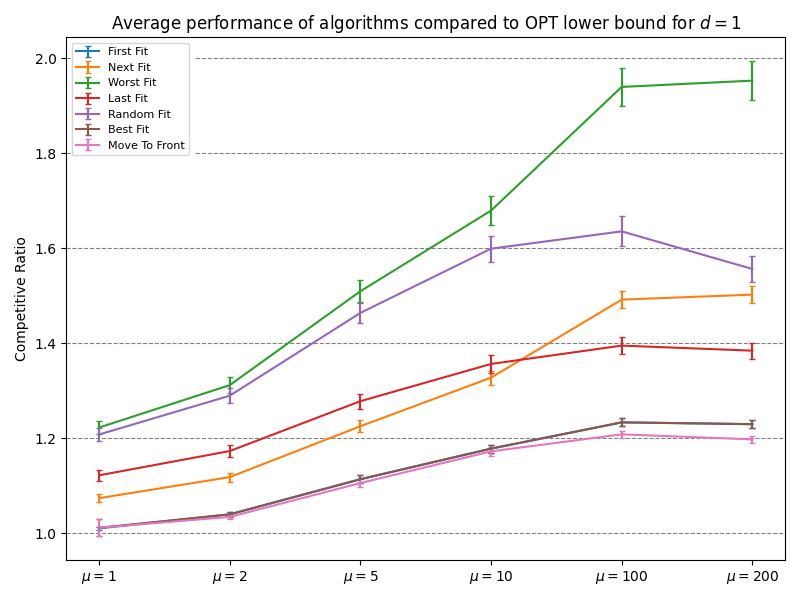}
  \caption{\normalfont $d=1$}
   \label{fig:d1} 
\end{subfigure}

\begin{subfigure}[b]{0.47\textwidth}
   \includegraphics[width=\textwidth]{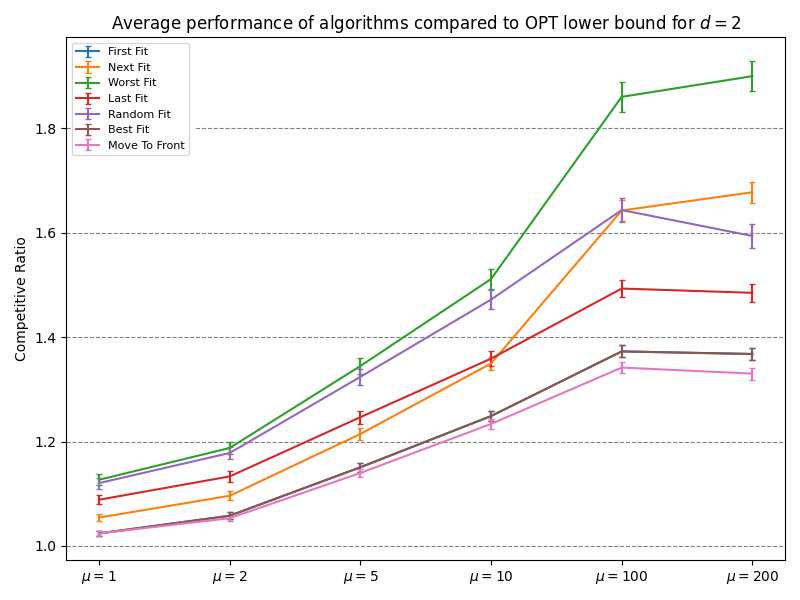}
  \caption{\normalfont $d=2$}
   \label{fig:d2}
\end{subfigure}

\begin{subfigure}[b]{0.47\textwidth}
   \includegraphics[width=\textwidth]{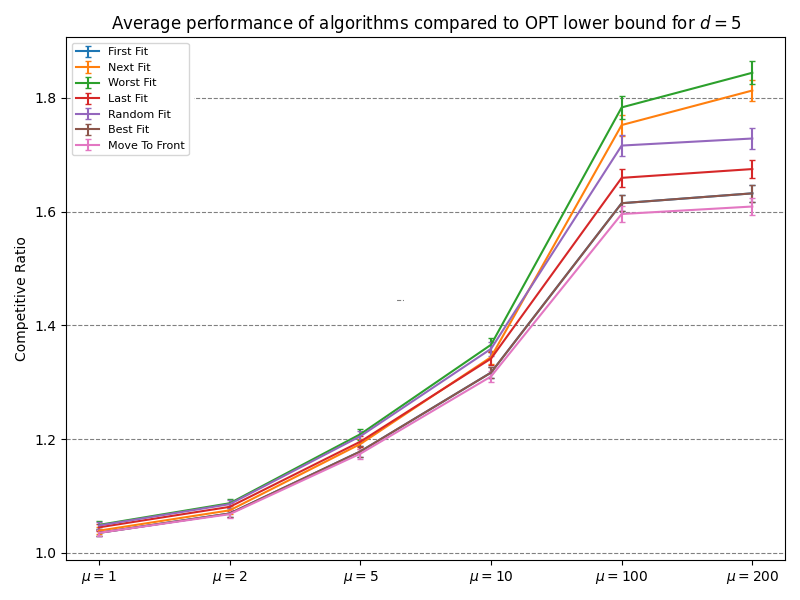}
  \caption{\normalfont $d=5$}
   \label{fig:d5}
\end{subfigure}
\caption{\normalfont Average-case performance of Any Fit packing algorithms for different values of $\mu$ and $d$. Error bars measure std. deviation. }\label{fig:expts}
\end{figure}

\paragraph{General Observations.} We observe that Move To Front has the best average-case performance among all algorithms, even in multiple dimensions. Close in performance are First Fit and Best Fit, which have nearly identical performance (the blue and brown curves are nearly superimposed), with First Fit having generally lower variance. Following these are Next Fit, Last Fit and Random Fit with the performance of Next Fit degrading with higher values of $\mu$, i.e., longer jobs on average. As expected, Worst Fit has the worst performance since it packs items inefficiently. Random Fit and Worst Fit are also seen to have a higher variance indicating a high variability in performance. In contrast, Move To Front, First Fit and Best Fit have low variance in performance.

\paragraph{Packing and Alignment.} We attempt to intuitively explain our experimental results. As \cite{KamaliLopezOrtiz2015} discuss, the quality of a solution is influenced by \textit{packing} and \textit{alignment}. Packing refers to how tightly items are packed together and how much space is wasted, while alignment refers to how well-aligned items are in terms of their durations. Better packing leads to improved performance since lower number of bins need to be opened. Better alignment saves on cost because if items arriving at almost the same time are packed together then in expectation they all depart together, preventing solutions in which multiple bins are active, each containing a small number of long-duration jobs.

The performance of Best Fit (resp. Worst Fit) can therefore be explained due to its excellent (resp. poor) packing. Next Fit generally should lead to well-aligned solutions since it tries to fit items into one bin, however does not factor packing into consideration as it only maintains one open bin. On the other hand, Move To Front does relatively well on both fronts: by using the most recently-used bin it leads to well-aligned solutions, and since it keeps all bins open it does not open many new bins like Next Fit.  These intuitive ideas can also be used to explain the observation that the performance of Next Fit worsens with large $\mu$: for larger $\mu$, it is more likely that bins remain open for a longer duration, hence it is better to avoid opening many new bins (which is what Next Fit does) to save cost. 

\paragraph{Theory vs Practice.} Our work invites an interesting discussion contrasting theory and practice. While the competitive ratio (CR) of Best Fit is theoretically unbounded, it has a good average-case performance in practice. In contrast, although the CR of Next Fit is theoretically bounded, it does not do as well in practice. Finally, although the bound on the CR of First Fit is lower than the proved lower bound on the CR of Move To Front, on average Move To Front has better performance than First Fit. These observations suggest theoretically studying the average-case performance of these algorithms against input sequences arising from specific distributions (such as the uniform distribution considered in our experiments) as an interesting direction for future research.

\section{Concluding Remarks}\label{sec:conclusion}
In this paper, we studied the MinUsageTime Dynamic Vector Bin Packing problem, where the size of an item is a $d$-dimensional vector, modelling multi-dimensional resources like CPU/GPU, memory, bandwidth, etc. We proved almost-tight lower and upper bounds on the competitive ratio (CR) of Any Fit packing algorithms such as Move To Front, First Fit and Next Fit. Notably, we showed that Move To Front has a CR of $(2\mu+1)d+1$, thus significantly improving the previously known upper bound of $6\mu+7$ for the 1-D case. Our experiments show that Move To Front has superior average-case performance than other Any Fit packing algorithms.

We discuss some interesting directions for future work. The first is to close the gap between the upper and lower bounds presented in this paper. Concretely, investigating if the lower bound of $\max\{2\mu, (\mu+1)d\}$ on the CR of Move To Front can be improved to $2\mu d$ is a natural first step.
Another direction is to design algorithms with improved CR for small number of dimensions, such as $d=2$. 
Lastly, studying the problem when given additional information about the input, perhaps obtained using machine learning algorithms, is another direction for future work. For instance, studying the clairvoyant DVBP problem, where the duration of an item is accurately known at the time of its arrival, is an interesting question.  

\bibliography{references}

\appendix
\section{Missing Proofs}
\subsection{Proof of Proposition~\ref{prop:norm}}\label{app:prelim}
\propnorm*
\begin{proof}
Part (i) and the first inequality of part (ii) are immediate from the definition of $L_\infty$ norm. The second inequality of part (ii) follows form the following series of inequalities.
\begin{equation*}
\begin{aligned}
\sum_{i=1}^n \norm{\vect{v}_i}_\infty &\le \sum_{i=1}^n\sum_{j=1}^d (\vect{v}_i)_j &\text{(since $\norm{\vect{v}_i}_\infty = \max_{j\in [d]} (\vect{v}_i)_j$)}\\
&= \sum_{j=1}^d\sum_{i=1}^n (\vect{v}_i)_j &\text{(changing the order of sum)}\\
&\le d\cdot\max_{j\in [d]}\sum_{i=1}^n (\vect{v}_i)_j \\
&= d \cdot \bigg\lVert{\sum_{i=1}^n \vect{v}_i}\bigg\rVert_\infty. &\text{(by definition of norm)}
\end{aligned}
\end{equation*}
\end{proof}

\end{document}